\documentclass[11pt,a4paper]{article}
\usepackage[utf8]{inputenc}
\usepackage{amsmath,amsthm,amssymb,mathtools}
\usepackage{algorithm,algorithmic}
\usepackage{microtype}
\usepackage{cite}
\usepackage{fullpage}
\usepackage{enumitem}
\usepackage{xcolor}
\usepackage{authblk}
\usepackage{booktabs}
\usepackage[small]{caption}
\usepackage{hyperref}

\DeclareMathOperator{\polylog}{polylog}

\newcommand{\cO}{O}
\newcommand{\A}{{\color{black} 1}}
\newcommand{\B}{{\color{blue} 2}}
\newcommand{\C}{{\color{red} 3}}
\newcommand{\D}{{\color{green!50!black} 4}}

\hypersetup{
    colorlinks=true,
    linkcolor=black,
    citecolor=black,
    filecolor=black,
    urlcolor=[rgb]{0,0.1,0.5},
    pdftitle={Distributed Recoloring},
    pdfauthor={Marthe Bonamy, Paul Ouvrard, Mika\"el Rabie, Jukka Suomela, Jara Uitto}
}

\title{Distributed Recoloring}
\author[1]{Marthe Bonamy}
\author[2]{Paul Ouvrard}
\author[3]{Mika\"el Rabie}
\author[4]{Jukka Suomela}
\author[5]{Jara Uitto}
\affil[1]{CNRS, LaBRI, Universit\'e de Bordeaux, marthe.bonamy@u-bordeaux.fr}
\affil[2]{LaBRI, CNRS, Universit\'e de Bordeaux, paul.ouvrard@u-bordeaux.fr}
\affil[3]{Aalto University, mikael.rabie@aalto.fi}
\affil[4]{Aalto University, jukka.suomela@aalto.fi}
\affil[5]{ETH Z\"urich \& University of Freiburg, jara.uitto@inf.ethz.ch}
\date{}

\newtheorem{theorem}{Theorem}[section]
\newtheorem{lem}[theorem]{Lemma}

\theoremstyle{definition}
\newtheorem{defn}[theorem]{Definition}

\begin{document}
\maketitle

\begin{abstract}
\normalsize
Given two colorings of a graph, we consider the following problem: can we recolor the graph from one coloring to the other through a series of elementary changes, such that the graph is properly colored after each step?

We introduce the notion of \emph{distributed recoloring}: The input graph represents a network of computers that needs to be recolored. Initially, each node is aware of its own input color and target color. The nodes can exchange messages with each other, and eventually each node has to stop and output its own recoloring schedule, indicating when and how the node changes its color. The recoloring schedules have to be globally consistent so that the graph remains properly colored at each point, and we require that adjacent nodes do not change their colors simultaneously.

We are interested in the following questions: How many communication rounds are needed (in the deterministic LOCAL model of distributed computing) to find a recoloring schedule? What is the length of the recoloring schedule? And how does the picture change if we can use \emph{extra colors} to make recoloring easier?

The main contributions of this work are related to distributed recoloring with one extra color in the following graph classes: trees, $3$-regular graphs, and toroidal grids.
\end{abstract}

\section{Introduction}

In classical graph problems, we are given a graph and the task is to \emph{find} a feasible solution. In \emph{reconfiguration problems}, we are given two feasible solutions -- an input configuration and a target configuration -- and the task is to find a sequence of moves that turns the input configuration into the target configuration.

\subparagraph{Recoloring problems.}
Perhaps the most natural example of a reconfiguration problem is \emph{recoloring}: we are given a graph $G$ and two proper $k$-colorings of $G$, let us call them $s$ and $t$, and the task is to find a way to turn $s$ into $t$ by changing the color of one node at a time, such that each intermediate step is a proper coloring. More formally, the task is to find a sequence of proper $k$-colorings $x_0, x_1, \dotsc, x_L$ such that $x_0 = s$ and $x_L = t$, and $x_{i-1}$ and $x_i$ differ only at one node. Such problems have been studied extensively from the perspective of graph theory and classical centralized algorithms, but the problems are typically inherently \emph{global} and solutions are long, i.e., $L$ is large in the worst case.

In this work we introduce recoloring problems in a \emph{distributed} setting. We show that there are natural relaxations of the problem that are attractive from the perspective of distributed graph algorithms: they admit solutions that are short and that can be found \emph{locally} (e.g., in sublinear number of rounds). Distributed recoloring problems are closely related to classical symmetry-breaking problems that have been extensively studied in the area of distributed graph algorithms, but as we will see, they also introduce new kinds of challenges.

\begin{figure}
	\centering
	\includegraphics[scale=0.8]{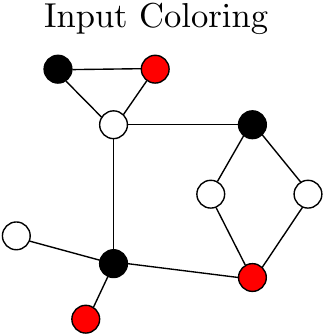}
	\quad
	\includegraphics[scale=0.8]{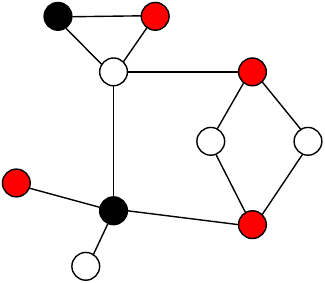}
	\quad
	\includegraphics[scale=0.8]{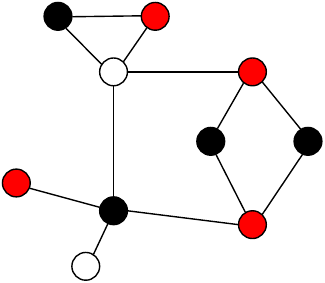}
	\quad
	\includegraphics[scale=0.8]{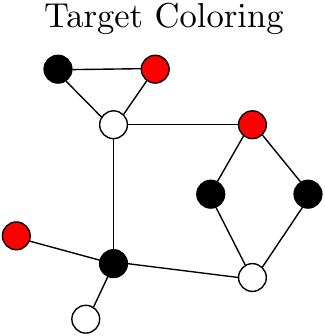}
	\caption{Distributed recoloring: the input coloring $s$ can be seen on the left and the target coloring $t$ on the very right. The illustration shows one possible way to reach the target coloring in three steps by, in each step, changing the colors of an independent set of nodes.}
	\label{fig:recolor}
\end{figure}

\subparagraph{Distributed recoloring.}

We will work in the usual LOCAL model of distributed computing: Each node $v \in V$ of the input graph $G = (V,E)$ is a computer, and each edge $e \in E$ represents a communication link between two computers. Computation proceeds in synchronous rounds: each node sends a message to each of its neighbors, receives a message from each of its neighbors, and updates its local state. Eventually, all nodes have to announce their local outputs and stop; the running time of the algorithm is the number of communication rounds until all nodes stop. We assume that the algorithm is deterministic, and each node is labeled with a unique identifier.

In \emph{distributed recoloring}, each node $v \in V$ is given two colors, an \emph{input color} $s(v)$ and a \emph{target color} $t(v)$. It is guaranteed that both $s$ and $t$ form a proper coloring of $G$, that is, $s(u) \ne s(v)$ and $t(u) \ne t(v)$ for all $\{u,v\} \in E$. Each node $v \in V$ has to output a finite \emph{recoloring schedule} $x(v) = \bigl(x_0(v), x_1(v), \dotsc, x_\ell(v)\bigr)$ for some $\ell = \ell(v)$. For convenience, we define $x_i(v) = x_\ell(v)$ for $i > \ell(v)$. We say that the node \emph{changes its color at time $i > 0$} if $x_{i-1}(v) \ne x_i(v)$; let $C_i$ be the set of nodes that change their color at time $i$. Define $L = \max_v \ell(v)$; we call $L$ the \emph{length} of the solution. A solution is feasible if the following holds:
\begin{enumerate}
\item $x_0 = s$ and $x_L = t$,
\item $x_i$ is a proper coloring of $G$ for all $i$,
\item $C_i$ is an independent set of $G$ for all $i$.
\end{enumerate}
The key differences between distributed recoloring and classical recoloring are:
\begin{enumerate}
\item Input and output are given in a distributed manner: no node knows everything about $G$, $s$, and $t$, and no node needs to know everything about $x_i$ or the length of the solution $L$.
\item We do not require that only one node changes its color; it is sufficient that adjacent nodes do not change their colors simultaneously.
\end{enumerate}
See Figure~\ref{fig:recolor} for a simple example of distributed recoloring steps.

Note that a solution to distributed recoloring is locally checkable in the following sense: to check that a solution is feasible, it is enough to check independently for each edge $\{u,v\} \in E$ that the recoloring sequences $x(u)$ and $x(v)$ are compatible with each other, and for each node $v \in V$ that $x(v)$ agrees with $s(v)$ and $t(v)$. However, distributed recoloring is not necessarily an LCL problem \cite{naor1995can} in the formal sense, as the length of the output per node is not a priori bounded.

We emphasize that we keep the following aspects well-separated: what is the complexity of \emph{finding} the schedule, and how \emph{long} the schedules are. Hence it makes sense to ask, e.g., if it is possible to find a schedule of length $\cO(1)$ in $\cO(\log n)$ rounds (note that the physical reconfiguration of the color of the node may be much slower than communication and computation).

\subparagraph{Recoloring with extra colors.}

Recoloring is computationally very hard, as solutions do not always exist, and deciding whether a solution exists is PSPACE-hard. It is in a sense analogous to problems such as finding an \emph{optimal} node coloring of a given graph; such problems are not particularly interesting in the LOCAL model, as the complexity is trivially global. To make the problem much more interesting we slightly relax it.

We define a \emph{$k+c$ recoloring problem} (a.k.a.\ \emph{$k$-recoloring with $c$ extra colors}) as follows:
\begin{itemize}
\item We are given colorings with $s(v), t(v) \in [k]$.
\item All intermediate solutions must satisfy $x_i(v) \in [k+c]$.
\end{itemize}
Here we use the notation $[n] = \{1,2,\dotsc,n\}$.

\begin{figure}[b]
	\centering
	\includegraphics{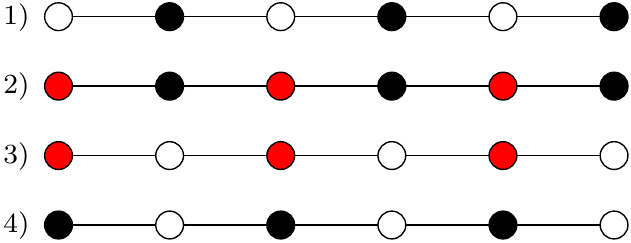}
	\caption{In the input graph, a bipartition is given. Therefore, the target coloring can be reached by using one extra color in three steps.}
	\label{fig:line}
\end{figure}

The problem of $k+c$ recoloring is meaningful also beyond the specific setting of distributed recoloring. For example, here is an example of a very simple observation:
\begin{lem}\label{lem:bipartite}
Recoloring with $1$ extra color is always possible in any bipartite graph, with a distributed schedule of length $L = 3$.
\end{lem}
\begin{proof}
Let the bipartition be $V = V_1 \cup V_2$. First each node $v \in V_1$ switches to $k+1$, then each $v \in V_2$ switches to color $t(v)$, and finally each $v \in V_1$ switches to color $t(v)$.
\end{proof}
Incidentally, it is easy to extend this result to show that $k$-recoloring with $c = \chi-1$ extra colors is always possible with a schedule of length $\cO(c)$ in a graph with chromatic number $\chi$, and in particular $k$-recoloring with $c = k-1$ extra colors is trivial.
Figure~\ref{fig:line} gives an illustration of recoloring a bipartite graph with one extra color.

As a corollary, we can solve distributed $k+1$ recoloring in trees in $\cO(n)$ rounds, with a schedule of length $\cO(1)$: simply find a bipartition and apply the above lemma. However, is this optimal? Clearly finding a bipartition in a tree requires $\Omega(n)$ rounds, but can we solve recoloring with $1$ extra color strictly faster?

These are examples of problems that we study in this work. We initiate the study of distributed complexity of recoloring, with the ultimate objective of finding a complete characterization of graph families and parameters $k$, $c$, and $L$ such that distributed $k+c$ recoloring with schedules of length $L$ can be solved efficiently in a distributed setting.

As we will see, the problem turns out to be surprisingly rich already in very restricted settings such as grids or $3$-regular trees. Many of the standard lower bound techniques fail; in particular, known results on the hardness of graph coloring do not help here, as we are already given two proper colorings of the input graph.

\subparagraph{Contributions.}

Our main contribution is a comprehensive study of the complexity of distributed recoloring in various graph families; the results are summarized in Tables \ref{tab:cycles}--\ref{tab:3reg}. The highlights of this work are the following results:
\begin{enumerate}
    \item \textbf{\boldmath An algorithm for $3+1$ recoloring on trees.} On trees, $3$-recoloring is inherently global: it is easy to see that the worst-case running time is $\Theta(n)$ and the worst-case schedule length is $\Theta(n)$. With one extra color, we can trivially find a schedule of length $\cO(1)$ in time $\cO(n)$. However, we show that we can do much better: it is possible to find a schedule of length $\cO(1)$ in time $\cO(\log n)$.
    
    Here the key component is a new algorithm that solves the following sub-problem in $\cO(\log n)$ rounds: given a tree, find an independent set $I$ such that the removal of $I$ splits the tree in components of size $1$ or $2$. This subroutine may find applications in other contexts as well.
    
    These results are presented in Section~\ref{sec:treepositive}.
    \item \textbf{\boldmath An algorithm for $3+1$ recoloring for graphs of degree at most $3$.} In general graphs, $3+1$ recoloring is not necessarily possible; we can construct a small $4$-regular graph in which $3+1$ recoloring is not solvable. However, we will show that if the maximum degree of the graph is at most $3$ (i.e., we have a \emph{subcubic} graph), $3+1$ recoloring is always possible. Moreover, we can find a schedule of length $\cO(\log n)$ in time $\polylog(n)$.
    
    This result is presented in Section~\ref{sec:subcubicpositive}.
    
    \item \textbf{\boldmath Complexity of $3+1$ recoloring on toroidal grids.}
    We also give a complete characterization of $3+1$ recoloring in one particularly interesting family of $4$-regular graphs: $2$-dimensional toroidal grids (a.k.a.\ torus grid graphs, Cartesian graph products of two cycles). While the case of $1$-dimensional grids (cycles) is easy to characterize completely, the case of $2$-dimensional grids turns out to be much more interesting.
    
    Here our main contribution is the following graph-theoretic result: in an $h \times w$ toroidal grid, $3+1$ recoloring is possible for any input if and only if (i)~both $h$ and $w$ are even, or (ii)~$h = 4$, or (iii)~$w = 4$. In all other cases we can find $3$-colorings $s$ and $t$ such that $t$ is not reachable from $s$ even if we can use $1$ extra color.
    
    As a simple corollary, $3+1$ recoloring is inherently global from the perspective of distributed computing, and it takes $\Theta(n)$ rounds to solve even if we have the promise that e.g.\ $h$ and $w$ are even (and hence a schedule of length $\Theta(1)$ trivially exists).
    
    This result is presented in Section~\ref{sec:grids}.
\end{enumerate}
Additionally, several simple upper and lower bounds and corollaries are given in Sections \ref{sec:simpleub} and~\ref{sec:simplecor}.

\subparagraph{Motivation.}

As a simple application scenario, consider the task of reconfiguring a system of unmanned aerial vehicles. Here each node is an aircraft, the color corresponds to an altitude range, and an edge corresponds to a pair of aircraft whose paths might cross and hence need to be kept at different cruising altitudes to avoid collisions.

For each aircraft there are designated areas in which they can safely change their altitude. To reconfigure the entire system, we could take all aircraft to these areas simultaneously. However, this may be a costly maneuver.

Another possibility is to reserve a longer timespan during which a set $X$ of aircraft may change their altitudes, whenever they happen to be at convenient locations. Now if we let two aircraft $u, v \in X$ change their altitudes during the same timespan, we need to ensure that any intermediate configuration is safe, regardless of whether $u$ or $v$ happens to change its altitude first. Furthermore, we would like to complete reconfiguration in minimal time (short schedule), and we would like to waste precious airspace as little as possible and hence keep as few altitude levels as possible in reserve for reconfiguration (few extra colors).

This scenario -- as well as many similar scenarios, such as the task of reconfiguring the frequency bands of radio transmitters in a manner that never causes interference, even if the clocks are not perfectly synchronized -- give rise to the following variant of distributed recoloring that we call \emph{weak recoloring}: if two adjacent nodes $u$ and $v$ change their color simultaneously at time $i$, then
$\bigl\{ x_{i-1}(u), x_i(u) \bigr\} \cap \bigl\{ x_{i-1}(v), x_i(v) \bigr\} = \emptyset$,
that is, we have a proper coloring regardless of whether $u$ or $v$ changes its color first.

Let us now contrast weak recoloring with \emph{strong recoloring}, in which adjacent nodes never change colors simultaneously. Trivially, strong recoloring solves weak recoloring. But the converse is also true up to constant factors: if we have $k$ input colors and a solution to weak recoloring of length $L$, then we can also find a solution to strong recoloring of length $kL$. To see this, we can implement one weak recoloring step in $k$ strong recoloring substeps such that in substep $j$ nodes of input color $j$ change their colors.

As our focus is on the case of a small number of input colors, we can equally well study strong or weak recoloring here; all of our results hold for either of them. While weak recoloring is closer to applications, we present our results using strong recoloring, as it has a more convenient definition.

\section{Related work}

\subparagraph{Reconfiguration and recoloring.}

Recoloring, and more generally combinatorial reconfiguration has received attention over the past few years. Combinatorial reconfiguration problems consist of finding step-by-step transformations between two feasible solutions such that all intermediate results are also feasible. They model dynamic situations where a given solution is in place and has to be modified, but no disruption can be afforded. We refer the reader to the nice survey~\cite{Jansurvey} for a full overview, and focus here on node coloring as a reference problem. 

As mentioned earlier, we introduce distributed recoloring here, but centralized recoloring has been studied extensively before. Two main models are considered:
\begin{enumerate}
    \item \emph{Node recoloring:} at each step, we can recolor a node into a new color that does not appear on its neighborhood
    \item \emph{Kempe recoloring:} at each step, we can switch the colors in a bichromatic component (we operate a Kempe change).
\end{enumerate}
 
The usual questions are of the form: Given a graph $G$ and an integer $k$, are all its $k$-colorings equivalent (up to node or Kempe recolorings)? What is the complexity of deciding that? What is the maximum number of operations needed to go from to the other? 
 
All of those questions can also be asked for two specific $k$-colorings $s$ and $t$ of $G$. Are they equivalent (up to node or Kempe recolorings)? What is the complexity of deciding that? What is the maximum number of operations needed to go from $s$ to $t$ in $G$?
 
While the complexity of questions related to Kempe recoloring remains elusive, the problems related to node recoloring are typically PSPACE-hard~\cite{bonsma2009finding}. The related question of deciding equivalence when a bound on the length of an eligible recoloring sequence is given as part of the input has also been considered \cite{bonsma2014complexity}. We know that the maximum number of operations needed to go from one $3$-coloring to another in a tree is $\Theta(n)$~\cite{cereceda2011finding}. While $(\Delta+1)$-recoloring a graph with no node of degree more than $\Delta$ is not always possible, having $\Delta+2$ colors always suffices \cite{jerrum1995very}, and there are also meaningful results to obtain for the problem of $(\Delta+1)$-recoloring~\cite{feghali2016reconfigurations}. Two other settings have received special attention: characterizing fully when $3$-recoloring is possible \cite{cereceda2011finding,cereceda2009mixing}, and guaranteeing short reconfiguration sequences in the case of sparse graphs for various notions of sparse \cite{bonamy2013recoloring,bousquet2016fast}. 

Kempe changes were introduced in 1879 by Kempe in his attempted proof of the Four Color Theorem~\cite{kempe79}. Though this proof was fallacious, the Kempe change technique has proved useful in, for example, the proof of the Five Color Theorem and a short proof of Brooks' Theorem. Most works on the topic initially focused on planar graphs, but significant progress was recently obtained in more general settings. We know that all $k$-colorings of a graph with no node of degree more than $k$ are equivalent (w.r.t.\ Kempe changes), except in the case of one very specific graph: the $3$-prism~\cite{bonamy2015conjecture,feghali2017kempe,meyniel3}.

Note that some other variants have also been studied, perhaps most notably the question of how many nodes to recolor at once so that the graph can be recolored \cite{mcdonald2015connectedness}.

While we will not discuss Kempe recoloring in our work, we point out that recoloring with extra colors is closely connected to Kempe recoloring: Kempe recolorability implies recolorability with one extra color (while the converse is not true). Hence the negative results related to one extra color also hold for Kempe recoloring.

\subparagraph{Distributed graph coloring.}

Panconesi and Srinivasan~\cite{panconesi1995local} have used Kempe operations to design efficient distributed algorithms for graph coloring with $\Delta$ colors. Other than that we are not aware of prior work on distributed recoloring.
On the other hand, the literature on the standard distributed coloring is vast. The best overview on the topic is the book by Barenboim and Elkin \cite{barenboim2013distributed};
the most important recent developments include the following results.
There is a randomized $O\bigl(\log^* n + 2^{\sqrt{\log \log n}}\bigr)$ -time algorithm for $(\Delta + 1)$-coloring by Chang et al.~\cite{Chang2018}.
In the case of trees, the number of colors can be reduced to $\Delta$ with the cost of increasing the runtime to $O(\log_{\Delta} \log n)$~\cite{Chang2016}.
On the deterministic side, the best known $(\Delta + 1)$-coloring algorithm requires $O(\Delta^{3/4} \log \Delta + \log^* n)$ communication rounds~\cite{Barenboim2015}.
In the case of trees, the \emph{rake-and-compress} -method by Miller and Reif gives a $3$-coloring in time $O(\log n)$~\cite{MillerR89}.

However, there seems to be surprisingly little technology that one can directly transfer between the coloring domain and recoloring domain. Toroidal grids are a good example: by prior work \cite{brandt2017lcl}, $3$-coloring is an inherently global problem, and by the present work, $3+1$ recoloring is an inherently global problem, but the arguments that are used in these proofs are very different (despite the fact that both of them are related to the idea that a ``parity'' is preserved).

\section{Preliminaries}
In this article, each graph $G=(V,E)$ is a simple undirected graph where $V$ represents its node set and $E$ its edge set. For a subset of nodes $S \subseteq V$, we denote by $G[S]$ the subgraph induced by $S$. For a node $u \in V$, we denote by $N(u)$ the \emph{open neighborhood} of $u$ that is the set of all the neighbors of $u$ and by $N[u]$ its \emph{closed neighborhood} i.e.\ the set $N(u) \cup \{u\}$. For a subset $S \subseteq V$, its closed neighborhood corresponds to the set $\bigcup_{u \in S} N[u]$.

The \emph{degree} of a node is the number of neighbors. A \emph{$k$-regular graph} is a graph in which all nodes have degree $k$, a \emph{cubic graph} is the same thing as a $3$-regular graph, and a \emph{subcubic graph} is a graph in which all nodes have degree at most $3$. A \emph{tree} is a connected acyclic graph, and a \emph{$k$-regular tree} is a tree in which each node has degree $1$ or $k$.

A \emph{maximal independent set} (MIS) $S \subseteq V$ is an independent set (i.e.\ a set of  pairwise non-adjacent nodes) such that for each non-MIS node $u \notin S, N(u) \cap S \neq \emptyset$. 

Given a graph $G=(V,E)$, a \emph{list-assignment} is a function which assigns to each node $v \in V$ a list of colors $L(v)$. An \emph{$L$-coloring} of $G$ is a function $c$ that assigns to each node $v \in V$ a color $c(v) \in L(v)$ such that for any two adjacent nodes $u, v \in V$, we have $c(u) \neq c(v)$. A graph $G$ is \emph{$k$-list-colorable} if it admits an \emph{$L$-coloring} for every list-assignment where the list of each node is of size at least $k$. Therefore, list-coloring generalizes node-coloring if we consider the special case where each node receives the same input list. The notion of \emph{$L$-recoloring} is the natural generalization of $k$-recoloring: the same elementary steps are considered, and every intermediate coloring must be an $L$-coloring. 

In order to output a recoloring schedule, it is convenient to consider the question of recoloring a graph $G$ from a coloring $s$ to a coloring $t$, rather than the more symmetric question of whether the two colorings are equivalent in the given setting. We take this opportunity to note that we can reverse time and hence recoloring schedule from $s$ to $t$ also yields a recoloring schedule from $t$ to $s$. In the rest of the paper, we therefore address the two questions as one.

\section{Warmup -- simple results}\label{sec:simpleub}

We will start by presenting a number of simpler upper and lower bounds that also serve as an introduction to the topic of distributed recoloring.

\subsection{Upper bounds}

\begin{lem}\label{lem:minusone}
In any graph, $k+c$ recoloring for $c = k-1$ is possible in $0$ communication rounds, with a schedule of length $\cO(k)$.
\end{lem}
\begin{proof}
Generalize the idea of Lemma~\ref{lem:bipartite}; note that the schedule of node $v$ depends only on $s(v)$ and $t(v)$, and not on the colors of any other node around it.
\end{proof}

\begin{lem}\label{lem:3paths}
In paths and trees, $3$-recoloring is possible in $\cO(n)$ rounds, with a schedule of length $\cO(n)$.
\end{lem}
\begin{proof}
Every node has full knowledge of the graph. The statement can be intuited by induction on the size of the tree, but we delay a formal proof to Section \ref{sec:treepositive} and more precisely Lemma \ref{lem:recoltree}.
\end{proof}

\begin{lem}\label{lem:3plus1paths}
In cycles and paths, $3+1$ recoloring is possible in $\cO(1)$ rounds, with a schedule of length $\cO(1)$.
\end{lem}
\begin{proof}
Use the input coloring to find a maximal independent set $I$. Nodes of $I$ switch to color $4$. Nodes of $V\setminus I$ induce paths of length $\cO(1)$, apply Lemma~\ref{lem:3paths} there to recolor each of the paths by brute force, without using the extra color $4$. Finally, nodes of $I$ switch to their target colors.
\end{proof}

\begin{lem}\label{lem:beyonddelta}
Let $G$ be a graph of maximum degree at most $\Delta$, and let $k \ge \Delta+2$. Then $k$-recoloring with $c$ extra colors is at least as easy as $(k-1)$-recoloring with $c+1$ extra colors.
\end{lem}
\begin{proof}
Given a $k$-coloring $x$, we can construct a $(k-1)$-coloring $x'$ as follows: all nodes of color $k$ pick a new color from $\{1,2,\dotsc,k-1\}$ that is not used by any of their neighbors. Note that $x \to x'$ is a valid step in distributed recoloring (nodes of color $k$ form an independent set), and by reversing the time, also $x' \to x$ is a valid step.

Hence to recolor $s \to t$ with $c$ extra colors, it is sufficient to recolor $s' \to t'$ with $c+1$ extra colors (color $k$ no longer appears in the input and target colorings and can be used as an auxiliary color during recoloring). Then we can put everything together to form a recoloring schedule $s \to s' \to t' \to t$, with only constant overhead in the running time and schedule length.
\end{proof}

\begin{lem}\label{lem:4plus1subcubic}
In subcubic graphs, $4+1$ recoloring is possible in $O(1)$ rounds, with a schedule of length $O(1)$.
\end{lem}
\begin{proof}
Use the input coloring to find a maximal independent set $I$ in constant time. Nodes of $I$ switch to color $5$. Delete $I$; we are left with a graph $G'$ that consists of paths and isolated nodes. Apply Lemmas \ref{lem:beyonddelta} and \ref{lem:3plus1paths} to solve $4+0$ recoloring in each connected component of $G'$. Finally nodes of $I$ can switch to their target colors.
\end{proof}

\begin{lem}\label{lem:4plus2grids}
In toroidal grids, $4+2$ recoloring is possible in $O(1)$ rounds, with a schedule of length $O(1)$.
\end{lem}
\begin{proof}
Pick a maximal independent set $I$, color it with color $6$, and delete; we have a graph of degree at most $3$ and $1$ extra color. Apply Lemma~\ref{lem:4plus1subcubic} to recolor it, and finally nodes of $I$ can switch to their target colors.
\end{proof}

\begin{lem}\label{lem:5plus1grids}
In toroidal grids, $5+1$ recoloring is possible in $O(1)$ rounds, with a schedule of length $O(1)$.
\end{lem}
\begin{proof}
Pick a maximal independent set $I$, color it with color $6$, and delete; we have a graph of degree at most $3$ and $5+0$ colors remaining. Apply Lemma~\ref{lem:beyonddelta} to reduce to the case of $4+1$ colors, and then use Lemma~\ref{lem:4plus1subcubic}.
\end{proof}

\begin{lem}\label{lem:MISplusforest}
For any graph $G$ on $n$ nodes, for any two $k$-colorings $\alpha, \beta$ of $G$, if we can compute in $\cO(f(n))$ rounds an MIS $S$ such that $V\setminus S$ induces a forest of trees of depth at most $\cO(d(n))$,
we can compute in $\cO(f(n)+d(n))$ rounds  how to $(k+1)$-recolor $G$ from $\alpha$ to $\beta$ with schedule of length $\cO(d(n))$.
\end{lem}
\begin{proof}
The idea is quite simple: each node in $S$ goes into color $k+1$. We then use the algorithm described in the proof of Lemma \ref{lem:recoltree} to find a recoloring with schedule of
length $\cO(d(n))$ for each connected component after the removal of $S$. After that, each node of $S$ can go to their final color.
\end{proof}

\subsection{Lower bounds}

\begin{lem}\label{lem:needsextra}
Recoloring without any extra colors is not possible in the following settings for some pairs of input and target colorings:
\begin{enumerate}[label=(\alph*),noitemsep]
\item $2$-recoloring paths or trees.
\item $2$-recoloring cycles.
\item $3$-recoloring cycles.
\item $2$-recoloring toroidal grids.
\item $3$-recoloring toroidal grids.
\item $4$-recoloring toroidal grids.
\item $5$-recoloring toroidal grids.
\item $2$-recoloring cubic graphs.
\item $3$-recoloring cubic graphs.
\item $4$-recoloring cubic graphs.
\end{enumerate}
\end{lem}
\begin{proof}
We can construct a source coloring in which no node can make a move, and a target coloring different from the input coloring. Here we show examples of the source coloring $s$; the target coloring can be constructed by $t(v) \equiv s(v) + 1 \bmod k$:
\begin{enumerate}[label=(\alph*)]
\item A path with $2$ nodes, $s = \begin{bmatrix}1 & 2\end{bmatrix}$.
\item A $4$-cycle, $s = \begin{bmatrix}1 & 2 & 1 & 2\end{bmatrix}$.
\item A $4$-cycle, $s = \begin{bmatrix}1 & 2 & 3\end{bmatrix}$.
\item A $4 \times 4$ grid,
$
  s = \begin{bsmallmatrix}
  1 & 2 & 1 & 2 \\
  2 & 1 & 2 & 1 \\
  1 & 2 & 1 & 2 \\
  2 & 1 & 2 & 1
  \end{bsmallmatrix}
$.
\item A $3 \times 3$ grid,
$
  s = \begin{bsmallmatrix}
  1 & 2 & 3 \\
  2 & 3 & 1 \\
  3 & 1 & 2
  \end{bsmallmatrix}
$.
\item A $4 \times 4$ grid,
$
  s = \begin{bsmallmatrix}
  1 & 2 & 3 & 4 \\
  3 & 4 & 1 & 2 \\
  1 & 2 & 3 & 4 \\
  3 & 4 & 1 & 2
  \end{bsmallmatrix}
$.
\item A $5 \times 5$ grid,
$
  s = \begin{bsmallmatrix}
  1 & 2 & 3 & 4 & 5 \\
  3 & 4 & 5 & 1 & 2 \\
  5 & 1 & 2 & 3 & 4 \\
  2 & 3 & 4 & 5 & 1 \\
  4 & 5 & 1 & 2 & 3
  \end{bsmallmatrix}
$.
\item Complete bipartite graph $K_{3,3}$, with $s$ constructed from the bipartition.
\item Prism graph: connect the nodes of a $3$-cycle colored with $\begin{bmatrix}1 & 2 & 3\end{bmatrix}$ to another $3$-cycle colored with $\begin{bmatrix}2 & 3 & 1\end{bmatrix}$, in this order.
\item Complete graph $K_4$. \qedhere
\end{enumerate}
\end{proof}

\begin{lem}\label{lem:3pathslb}
In paths and trees, $3$-recoloring without extra colors requires $\Omega(n)$ rounds and produces schedules of length $\Omega(n)$ in the worst case. This holds also in the case of $3$-regular trees.
\end{lem}
\begin{proof}
Consider a long path with the input coloring $1,2,3,1,2,3,\dotsc,1,2,3$ and observe that a node of degree $2$ can change its color only after at least one neighbor has changed colors. We can embed such a path also in a $3$-regular tree.
\end{proof}

\begin{lem}\label{lem:4treelb}
In trees, $4$-recoloring without extra colors requires $\Omega(\log n)$ time and produces schedules of length $\Omega(\log n)$ in the worst case.
\end{lem}
\begin{proof}
It is sufficient to construct a $3$-regular tree in which each node is surrounded by nodes of all other colors: color the root with color $1$ and its neighbors with colors $2$, $3$, and $4$. Then recursively for each leaf node of color $x$ that is already adjacent to a node of color $y$, add two new neighbors with the colors $\{1,2,3,4\} \setminus \{x,y\}$, etc., and continue in a balanced manner such that the distance between the root and the nearest leaf is logarithmic. Now a non-leaf node can change its color only once its neighbor has changed its color.
\end{proof}

\section{Recoloring algorithm for trees}\label{sec:treepositive}

In this section, we provide two efficient algorithms for recoloring and list-recoloring trees. Note that Theorem~\ref{thm:tree4} is tight; see the full version for more details.

\begin{theorem}\label{thm:tree3-1}
For any $k\in \mathbb{N}$, for every tree $T$ on $n$ nodes, for any two $k$-colorings $\alpha, \beta$ of $T$, we can compute in $\cO(\log n)$ rounds how to recolor $T$ from $\alpha$ to $\beta$ with $1$ extra color and a schedule of length $\cO(1)$.
\end{theorem}

\begin{theorem}\label{thm:tree4}
For every tree $T$ on $n$ nodes and any list assignment $L$ of at least $4$ colors to every node of $T$, for any two $L$-colorings $\alpha, \beta$ of $T$, we can compute in $\cO(\log n)$ rounds how to $L$-recolor $T$ from $\alpha$ to $\beta$ with schedule of length $\cO(\log n)$.
\end{theorem}

We first discuss how to compute efficiently an independent set with some desirable properties. For this, we use a simple modification of the \emph{rake and compress} method by Reif and Miller~\cite{MillerR89}. More precisely, we iterate rake and compress operations, and label nodes based on the step at which they are reached. We then use the labels to compute an independent set satisfying given properties. We finally explain how to make use of the special independent set to obtain an efficient recoloring algorithm, in each case.

\begin{defn}\label{def:hlabel}
A \emph{light $h$-labeling} is a labeling $V \to [h]$ such that for any $i \in [h]$:
\begin{enumerate}
    \item Any node labeled $i$ has at most two neighbors with label $\geq i$, at most one of which with label $\geq i+1$.
    \item No two adjacent nodes labeled $i$ both have a neighbor with label $\geq i+1$.
\end{enumerate}
\end{defn}

\begin{lem}\label{lem:treelabeling}
There is an $O(\log n)$-round algorithm that finds a light $(2 \log n)$-labeling of a tree.
\end{lem}

\begin{proof}
As discussed above, we merely use a small variant of the \emph{rake and compress} method. At step $i$, we remove all nodes of degree $1$ and all nodes of degree $2$ that belong to a chain of at least three nodes of degree $2$, and assign them label $i$. 

One can check that this yields a light labeling. It remains to discuss how many different labels are used, i.e. how many steps it takes to delete the whole tree. Let us argue that no node remains after $2 \log n$ rounds. Let $T$ be a tree, let $V_1$ (resp.\ $V_2$, $V_3$) be the number of nodes of degree $1$ (resp.\ $2$, $\geq 3$) in the tree, and let $T'$ be the tree obtained from $T$ by replacing any maximal path of nodes of degree $2$ with an edge. Note that $|V(T')|=|V_1|+|V_3|$. Let $W$ be the set of nodes in $T$ that have degree $2$ with both neighbors of degree $2$. Note that $|V_2 \setminus W|\leq 2 |E(T')| =2(|V_1|+|V_3|-1)$. Note also that $|V_1|\geq |V_3|$, simply by the fact that there are fewer edges than nodes in a tree. It follows that $|W|\geq |V_2|- 2(|V_1|+|V_3|-1)= |V(T)|-|V_1|-|V_3|-2(|V_1|+|V_3|-1) \geq |V(T)|-6|V_1|$. Consequently, we obtain $|W|+|V_1|\geq \frac{|V|}6$. In other words, at every step, we remove in particular $W \cup V_1$, hence at least a sixth of the nodes. It follows that at after $k$ steps, the number of remaining nodes is at most $n \cdot \bigl(\frac56\bigr)^k$. Note that this is less than $1$ once $k \geq 2 \log n$.
\end{proof}

We now discuss how to make use of light $h$-labelings.

\begin{lem}\label{lem:uselightlabels}
For any graph $T$, any $3$-coloring $\alpha$ of $T$, and any integer $h$, let $L$ be a light $h$-labeling of $T$. There is an $O(h)$-round algorithm that finds a maximal independent set $S$ such that $T \setminus S$ only has connected components on $1$ or $2$ nodes.
\end{lem}

\begin{proof}
In brief, we proceed as follows: at step $i = h, h-1, \dotsc, 1$, we first add all nodes of label $i$ which have a neighbor of label $\geq i+1$ that is not in $S$ (they form an independent set by definition of a light label), then use the $3$-coloring to obtain a fast greedy algorithm to make $S$ maximal on the nodes of label $\geq i$. The detailed algorithm can be found in the full version.

\begin{algorithm}[t]
\caption{\label{algo:stabletreedecomposition}\textsc{Decomposing into an independent set and components of size $\le2$}}
\begin{algorithmic}[1] 
\REQUIRE A tree $T$, a 3-coloring $\alpha$ and a light $h$-label of $T$.
\ENSURE A set $S$ of $V(T)$ such that $G[S]$ is an independent set and every connected component of $G[V\setminus S]$ has size at most $2$.
\FOR{$i$ from $h$ down to 1}
\FOR{$u$ with label $i$ (in parallel)}
\STATE If $u$ has a neighbor of higher label that is not in $S$, add $u$ to $S$
\ENDFOR
\FOR{$j$ from $1$ to 3}
\FOR{$u$ with label $i$ and color $j$ (in parallel)}
\STATE If $N(u)\cap S=\emptyset$, add $u$ to $S$
\ENDFOR
\ENDFOR
\ENDFOR
\end{algorithmic}
\end{algorithm}

The fact that the output $S$ is an independent set follows directly from the construction, as does the fact that the running time in $O(h)$ rounds. We note that no connected component of $T \setminus S$ contains nodes of different labels, due to the first operation at step $i$. 

It remains to argue that for any $i$, the nodes of label $i$ that do not belong to $S$ only form connected components of size $1$ or $2$. Assume for a contradiction that there is a node $u$ of label $i$ which has two neighbors $v$ and $w$, also of label $i$, such that none of $\{u,v,w\}$ belongs to $S$. By definition of a light label, the node $u$ has no other neighbor of label $\geq i$, a contradiction to the fact that we build $S$ to be an MIS among the nodes of label $\geq i$.
\end{proof}

Combining Lemmas~\ref{lem:treelabeling} and~\ref{lem:uselightlabels}, and observing that a $3$-coloring of a tree can be obtained in $O(\log n)$ rounds, we immediately obtain the following.

\begin{lem}\label{lem:treeMIS}
There is an $O(\log n)$-round algorithm that finds an MIS in a tree, such that every component induced by non-MIS nodes is of size one or two.
\end{lem}

We are now ready to prove Theorem~\ref{thm:tree3-1}.

\begin{proof}[Proof of Theorem~\ref{thm:tree3-1}]
First, we use Lemma~\ref{lem:treeMIS} to obtain in $O(\log n)$ rounds an MIS $S$ such that $T \setminus S$ only has connected components of size $1$ or $2$. We recolor each node in $S$ with the extra color. Remove $S$, and recolor each component from $\alpha$ to $\beta$ without using any extra colors; this can be done in $O(1)$ recoloring rounds. Each node in $S$ can then go directly to its color in~$\beta$.
\end{proof}

Moving on to the list setting, we have to use a more convoluted approach since there is no global extra color that we can use. Before discussing $4$-list-recoloring, we discuss $3$-list-recoloring. For the sake of intuition, we start by presenting an algorithm for $3$-recoloring trees, and explain afterwards how to adapt it for the list setting.

\begin{lem}\label{lem:recoltree}
For every tree $T$ with radius at most $p$ and for any two 3-colorings $\alpha, \beta$ of $T$, we can compute in $O(p)$ rounds how to $3$-recolor $T$ from $\alpha$ to $\beta$ with a schedule of length $O(p)$.
\end{lem}

\begin{proof}
Let $c\colon V \to [3]$ be a $3$-coloring of $T$. We introduce an identification operation: Given a leaf $u$ and a node $v$ such that $u$ and $v$ have a common neighbor $w$, we recolor $u$ with $c(v)$, and from then on we pretend that $u$ and $v$ are a single node. In other words, we delete $u$ from the tree we are considering, and reflect any recoloring of $v$ to the node $u$. Note that these operations can stack up: the recoloring of a single node might be reflected on an arbitrarily large independent set in the initial tree. 

We now briefly describe an algorithm to recolor a $3$-coloring into a $2$-coloring $c'$ in $O(p)$ rounds, with schedule $O(p)$. First, root $T$ on a node $r$ which is at distance at most $p$ of any node of $T$. Any node of $T$ which is not adjacent to the root has a \emph{grandparent}, which is defined as its parent's parent.

Then, at each step, we consider the set $A$ of leaves of $T$ which have a grandparent, if any. We identify each leaf in $A$ with its grandparent (note that the notion of grandparent guarantees that this operation is well-defined, and that the operation results in $A$ being deleted).

This process stops when $T$ consists only of the root $r$ and its children. We select one of the children arbitrarily and identify the others with it. This results in $T$ being a single edge. Note that the color partition of $c'$ is compatible with the identification operations, as we only ever identify nodes at even distance of each other.

We then recolor $T$ into $c'$: this is straightforward in the realm of $3$-recoloring.

We can now choose a $2$-coloring of $T$ (this can be done in $O(p)$ rounds), and apply the above algorithm to $3$-recolor both $\alpha$ and $\beta$ to that $2$-coloring. This results in a $3$-recoloring between $\alpha$ and $\beta$ with schedule $O(p)$.
\end{proof}

The same idea can be adapted to list coloring:

\begin{lem}\label{lem:recoltreelist}
For every tree $T$ with radius at most $p$, for any list assignment $L$ of at least $3$ colors to each node, for any two $L$-colorings $\alpha, \beta$ of $T$, we can compute in $O(p)$ rounds how to $L$-recolor $T$ from $\alpha$ to $\beta$ with schedule $O(p)$.
\end{lem}

\begin{proof}
We adapt the identification operation introduced in the proof of  Lemma~\ref{lem:recoltree}, merely by adapting the notion of having the same color. Let $u$ and $v$ be two nodes with a common neighbor $w$. We say $u$ has the same color as $v$ \emph{with respect to $w$} in the following cases:
\begin{itemize}
    \item If $L(u)\neq L(w)$, then $u$ is colored with the smallest element of $L(u) \setminus L(w)$
    \item If $L(u)=L(w)$ and the color of $v$ belongs to $L(u)$, then $u$ is colored the same as $v$
    \item If $L(u)=L(w)$ and the color of $v$ does not belong to $L(u)$, then $u$ is colored with the smallest element of $L(w)$ that differs from the color of $w$
\end{itemize}

Therefore, when we identify a leaf $u$ with a node $v$ that has a common neighbor $w$ with $u$, we first assign to $u$ the same color as $v$ with respect to $w$, and from then on we pretend that $u$ and $v$ are a single node. In other words, any recoloring of $v$ is mirrored on $u$ so that at each step, the node $u$ has the same color as $v$ with respect to $w$. Note that in some cases it may be that the color of $u$ does not actually change when the color of $v$ does.

When the operations stack up, i.e. a node $u$ is identified with a node $v$ which is identified with a node $x$, we do not claim transitivity of the relation. In particular, $u$ and $x$ have no common neighbor, hence them having the same color is not well-defined. We merely enforce that $u$ has the same color as $v$ with respect to their common neighbor, and that $v$ has the same color as $w$ with respect to their common neighbor.

We insist on the fact that the definition of having the same color only depends on the list assignment. In particular, let us consider the situation once no more identification operation can be operated, i.e. the tree has been identified into an edge (see the proof of Lemma~\ref{lem:recoltree}). The coloring of the edge characterizes entirely the coloring of the whole tree, regardless of the initial coloring. Therefore, we can pick an arbitrary $L$-coloring of the edge, and recolor both $\alpha$ and $\beta$ into the corresponding $L$-coloring of the tree in $O(p)$ rounds with schedule $O(p)$.

This results in computing in $O(p)$ rounds an $L$-recoloring between $\alpha$ and $\beta$ with a schedule of length~$O(p)$.
\end{proof}

To prove Theorem~\ref{thm:tree4}, we first split the tree in small components. We slightly adapt the proof of Lemma~\ref{lem:uselightlabels}:

\begin{lem}\label{lem:uselightlabelsforlists}
For any tree $T$, any $3$-coloring $\alpha$ of $T$, and any integer $h$, let $L$ be a light $h$-label of $T$. There is a $O(h)$-round algorithm that finds a maximal independent set $S$ such that no node has two neighbors in $S$ and $T \setminus S$ only has connected components of radius $O(h)$.
\end{lem}

\begin{proof}
The algorithm is far simpler than Algorithm~\ref{algo:stabletreedecomposition}. We compute the set $R$ of nodes with no neighbor of higher label. We note that $G[R]$ is a collection of paths and cycles. We compute an independent set $S \subseteq R$ that is maximal subject to the property that no node in $R$ has two neighbors in $S$. Note that by definition of light label, no node outside of $R$ may have two neighbors in $R$ (hence in $S$). It remains to argue that $T \setminus S$ only has connected components of radius $O(h)$. We point out that every connected component of $T[R]$ contains an element of $S$. Therefore, any connected subset of nodes of $T[R]$ has at most one neighbor of higher label, since $T$ is a tree. Together with the fact that any connected component of $T[R \setminus S]$ has at most $2$ nodes, we derive the conclusion.
\end{proof}

Now we are ready to prove Theorem~\ref{thm:tree4}.

\begin{proof}[Proof of Theorem~\ref{thm:tree4}]
Compute (in $O(\log n)$ rounds) an independent set $S$ such any two elements of $S$ are at distance at least $2$ of each other and every connected component of $T \setminus S$ has radius $O(\log n)$. By Lemmas~\ref{lem:treelabeling} and~\ref{lem:uselightlabelsforlists} and the fact that a $3$-coloring of a tree can be computed in $O(\log n)$ rounds, we compute (in $O(\log n)$ rounds) an $L$-coloring $\gamma$ of $T \setminus S$ such that every node adjacent to an element $u \in S$ has a color different from $\alpha(u)$ and $\beta(u)$. Note that this coloring exists since any tree is $2$-list-colorable. Use Lemma~\ref{lem:recoltreelist} to recolor each connected component of $T \setminus S$ from $\alpha$ to $\gamma$. Recolor every element of $S$ with its color in $\beta$. Use Lemma~\ref{lem:recoltreelist} to recolor each connected component $T \setminus S$ from $\gamma$ to $\beta$. Note that this yields an $L$-recoloring of $T$ from $\alpha$ to $\beta$ with schedule $O(\log n)$.
\end{proof}

Note that a direct corollary of Theorem \ref{thm:tree4} is that for any $k-$coloring $\alpha$, $\beta$ of a trees with $k\ge4$, a schedule of length  $\Theta(\log n)$  can be found in $\Theta(\log n)$ rounds.

\section{Recoloring algorithm for subcubic graphs}\label{sec:subcubicpositive}

In this section we study recoloring in subcubic graphs (graphs of maximum degree at most $3$); our main result is summarized in the following theorem:

\begin{theorem}\label{thm:cubic3-1}
For every subcubic graph $G$ on $n$ nodes, for any two $3$-colorings $\alpha, \beta$ of $G$, we can compute in $O(\log^2 n)$ rounds how to recolor $G$ from $\alpha$ to $\beta$ with $1$ extra color and a schedule of length $O(\log n)$.
\end{theorem}

A \emph{theta} is formed of three node-disjoint paths between two nodes. Note that in particular if a graph contains two cycles sharing at least one edge, then it contains a theta. We note $B^{k}(u)$ the set of nodes at distance at most $k$ to $u$.

We show here, roughly, that there is around every node a nice structure that we can use to design a valid greedy algorithm for the whole graph. This proof is loosely inspired by one in~\cite{ABBE}.

\begin{lem}\label{lem:thereisthetaor2}
For every subcubic graph $G$ on $n$ nodes, for every node $u \in V(G)$, there is a node $v$ with degree at most $2$ or a theta that is contained in $B^{2 \log n}(u)$.
\end{lem}

\begin{proof}
Assume for a contradiction that there is a subcubic graph $G$ on $n$ nodes with a node $u$ such that $B^{2 \log n}(u)$ contains no node of degree $2$ nor any theta. Let $B$ be the set of nodes at distance at most $2 \log n$ from $u$, and $B^-$ the set of nodes at distance at most $2 \log n-1$ from $u$. Let $\mathcal{C}$ be the set of cycles of $G$ contained in $B$. Note that cycles in $\mathcal{C}$ are edge-disjoint by assumption on $u$ and thus node-disjoint since $G$ is cubic. We select a set $\mathcal{E}$ by picking for every $C \in \mathcal{C}$ an arbitrary edge in $E(C)$ among those with both endpoints farthest from $u$. Note that $|\mathcal{E}|=|\mathcal{C}|$, and that by choice of $\mathcal{E}$, every edge in $B$ with both endpoints at the same distance of $u$ is selected in~$\mathcal{E}$. Therefore, the distance to $u$ yields a natural orientation of the edges in $B \setminus \mathcal{E}$, orientation from closer node to $u$ toward further node. We also note that by choice of $\mathcal{E}$, for any edge $wx$ in $\mathcal{E}$ such that $x$ is farther away from $u$ than $w$, the node $x$ has another neighbor $y$ at the same distance of $u$ as $w$. In that case, note that the edge $xy$ does not belong to $\mathcal{E}$.
We claim as a consequence that the distance from $u$ is the same in $B$ as in $B \setminus \mathcal{E}$.

For any node $w \in B$, we say an outgoing edge is \emph{useful} if it does not belong to $\mathcal{E}$.
In addition to the above remarks, we make two observations:
\begin{enumerate}
    \item Every node in $B^-$ has at least one useful edge.
    \item If a node $w$ in $B^-$ has only one useful edge $wx$, then $x$ has two outgoing useful edges.
\end{enumerate}

Let us consider the graph $H$ obtained from $G[B]$ by removing all edges in $\mathcal{E}$. We claim that every node in $B$ has degree at least $2$ in $H$, and that no two adjacent nodes in $H$ have degree $2$: this is immediate from the observations and remarks above. We also observe that $H$ is a tree. Let $H'$ be the graph obtained from $H$ by replacing every node of degree two with an edge. We note that $H'$ is a $3$-regular tree of root $u$ and with no leaf at distance less than $\log n$ of $u$. It follows that $H'$ contains at least $1+3 \cdot 2^{\log n} > n$ nodes, a contradiction.
\end{proof}

\begin{lem}\label{lem:decompositionextension}
Let $G$ be a subcubic graph, let $p$ be an integer, and let $\mathcal{A}$ be a collection of thetas and nodes of degree $\leq 2$ in $G$ each at distance at least $2$ of each other. Let $r \geq 1$ be such that no element of $\mathcal{A}$ has diameter more than $\frac{r}2$. If the nodes of $G \setminus (\bigcup_{A \in \mathcal{A}} A)$ can be partitioned into $S$ and $F$ such that $G[S]$ is an independent set and $G[F]$ is a forest of radius at most $p$, then there is a partition $(S',F')$ of $\bigcup_{A \in \mathcal{A}} A$ such that $G[S \cup S']$ is an independent set and $G[F \cup F']$ is a forest of radius at most $p+r$.
\end{lem}

\begin{proof}
Our construction ensures that any pair of nodes that are not connected in $G[F]$ are not connected in $G[F \cup F']$ neither. Hence, it suffices to prove that the statement holds for a single element of $\mathcal{A}$, since the elements of $\mathcal{A}$ are by hypothesis non-adjacent. Let $A$ be an element of  $\mathcal{A}$. We consider two cases depending on whether $A$ is a node of degree at most $2$ or is a theta.

\begin{itemize}
    \item If $A$ consists of a node $v$ of degree $1$, or $2$, we set $v$ to be in $F'$ if it has a neighbor in $S$, in $S'$ otherwise. Note that since $v$ has at most one neighbor in $F$, the radius of $F \cup F'$ is at most one more than that of $F$.
    \item If $A$ consists of a theta with endpoints $u$ and $v$ and three node-disjoint paths $P_1, P_2, P_3$, we prove independently that each $P_i$ admits a partition that is compatible with $u$ being set to $S'$ and $v$ to $F'$, in such a way that the connected component of $F \cup F'$ that contains $v$ is contained in $F'$. We do this by induction on the number of nodes in $P_i$. If $P_i$ has no internal node, the conclusion immediately follows. If all the neighbors of $P_i$ at distance $\geq 3$ of $u$ through $P_i$ are in $S$, we set all of $P_i$ to $F'$. Otherwise, let $w$ be the neighbor of $P_i$ in $F$ that with smallest distance ($\geq 3)$ to $u$ through $P_i$. Let $x$ be the neighbor of $w$ in $P_i$. We apply induction on $P_i \setminus \{\textrm{nodes closer to $u$ than $x$}\}$, with $x$ in the role of $u$. Note that $x$ is distinct from $v$ and not adjacent to $u$, by construction. The nodes between $u$ and $x$ on $P_i$ are added to $F'$. Note that these nodes are connected to at most one component of $G[F]$, on the first node of $P_i$.
    We extend the resulting decomposition to the rest of $P_i$ by setting all corresponding nodes to $F'$. \qedhere
\end{itemize}
\end{proof}

\begin{lem}\label{lem:algostableforestdecompositionrunningtime}
Let $G$ be a subcubic graph on $n$ nodes. We can compute in $O(\log^2 n)$ rounds a partition $(S,F)$ of the nodes of $G$ that $G[S]$ is an independent set and $G[F]$ is a forest of radius $O(\log n)$.
\end{lem}

\begin{proof}
To that purpose, we combine the previous lemmas in Algorithm~\ref{algo:stableforestdecomposition}.
The algorithm computes a decomposition as desired and runs in $O(\log n)+RS(n)$ rounds, where $RS(n)$ is the number of rounds necessary to compute a $(4 \log n, 8 \log n)$-ruling set in a subcubic graph. We derive from~\cite{panconesi1992improved} that $RS(n) = O(\log^2 (n))$, hence the conclusion.
\end{proof}

\begin{algorithm}
\caption{\label{algo:stableforestdecomposition}\textsc{Decomposing into a small forest and an independent set}}
\begin{algorithmic}[1] 
\REQUIRE A subcubic graph $G$.
\ENSURE A decomposition $(F,S)$ of $V(G)$ such that $G[S]$ is an independent set and every connected component of $G[F]$ has radius at most $\log n$.
\FOR{$u$ in $V(G)$ (in parallel)}
\STATE Acquire knowledge on $B^{2 \log n}(u)$
\STATE Select in the node set of $B^{2 \log n}(u)$ a configuration $C(u)$ that is a minimal theta or a node of degree $1$ or $2$
\ENDFOR
\STATE Compute a $(4 \log n, 8 \log n)$-ruling set $X$ in $G$
\STATE Define $\mathcal{A}=\cup_{u \in X}\{C(u)\}$
\STATE Compute the distance of every node in $G$ to an element of $\mathcal{A}$
\STATE Let $F=S=\emptyset$
\FOR{$i=8 \log n$ downto $1$}
\STATE Extend the partition $(F,S)$ to the nodes at distance $i$ from $\mathcal{A}$, more precisely:
\STATE Each connected component is a path or cycle where no internal node has an already assigned neighbor, let $U_i$ be the set of the internal nodes
\STATE Assuming a pre-computed MIS on each layer for the sets $U_i$, assign that MIS to $S$
\STATE Extend greedily on the remaining nodes (which form bounded-size components), assigning nodes to $S$ when possible, to $F$ when not
\ENDFOR
\STATE Extend the partition $(F,S)$ to the nodes belonging to an element of $\mathcal{A}$ using Lemma~\ref{lem:decompositionextension}
\end{algorithmic}
\end{algorithm}

We are now ready to prove Theorem \ref{thm:cubic3-1}, which we do in a similar fashion as Theorem \ref{thm:tree3-1}.

\begin{proof}
Use Lemma \ref{lem:algostableforestdecompositionrunningtime}, and obtain a decomposition $(S,F)$ as stated. Recolor all of $S$ to the extra color, then use Lemma \ref{lem:recoltreelist} on each connected component of $G[F]$ so that all nodes of $F$ reach their target color (remember that each connected component of $G[F]$ has radius $O(\log n)$). Finally recolor each node of $S$ with its target color.
\end{proof}

\section{Recoloring in toroidal grids}\label{sec:grids}

In this section we study toroidal grids (torus grid graphs). Throughout this section, an $h \times w$ toroidal grid is the Cartesian graph product of cycles of lengths $h$ and $w$; we assume $h \ge 3$ and $w \ge 3$. A toroidal grid can be constructed from an $h \times w$ grid by wrapping both boundaries around into a torus. In the full version, we show that e.g.\ $2+0$, $3+0$, and $4+0$ recoloring is not always possible, and by Lemma~\ref{lem:minusone} e.g.\ $2+1$, $3+2$, and $4+3$ recoloring is trivial. The first nontrivial case is $3+1$ recoloring; in this section we give a complete characterization of $3+1$ recolorability in toroidal grids:

\begin{theorem}\label{thm:grids}
Let $G$ be the $h \times w$ toroidal grid graph. Then $3+1$ recoloring is possible for any source and target coloring in the following cases:
(i)~both $h$ and $w$ are even, or
(ii)~$h = 4$, or
(iii)~$w = 4$.
For all other cases it is possible to construct $3$-colorings $s$ and $t$ such that $t$ is not reachable from $s$ by valid recoloring operations using $1$ extra color.
\end{theorem}

This also shows that $3+1$ recoloring is an inherently global problem in toroidal grids, even if we have a promise that recoloring is possible. For example, if there was a sublinear-time distributed recoloring algorithm $A$ for $6 \times w$ grids for an even $w$, we could apply the same algorithm in a $6 \times w$ grid with an odd $w$ (the algorithm cannot tell the difference between these two cases in time $o(w)$), and hence we could solve recoloring in $6 \times w$ grids for all $w$, which contradicts Theorem~\ref{thm:grids}. By a similar argument, distributed recoloring in non-toroidal grids is also an inherently global problem.

\subparagraph{Existence.}

To prove Theorem~\ref{thm:grids}, let us start with the positive results. If $h$ and $w$ are even, the graph is bipartite and recoloring is always possible by Lemma~\ref{lem:bipartite}. The remaining cases are covered by the following lemma.
\begin{lem}\label{lem:grids4xw}
    Let $G$ be a $4 \times w$ toroidal grid for any $w \ge 3$, and let $s$ and $t$ be any $3$-colorings. Then there exists a recoloring from $s$ to $t$ with one extra color.
\end{lem}
\begin{proof}
We first take an MIS $S$ over pairs of consecutive columns, i.e.\ a set of indices of the form $(i,i+1)$ such that every column $j \notin S$ is such that at least one of $j-1$ and $j+2$ belongs to $S$, every column $i \in S$ is such that precisely one of $i-1$ and $i+1$ is in $S$. Note that indices are taken modulo $w$. 
For every pair in $S$, we select a maximal independent set of the corresponding columns. The resulting union yields an independent set $R$. We then greedily make $R$ maximal columnwise away from $S$. We recolor $R$ with the extra color. It remains to argue that $G\setminus R$ can reach its targeted coloring. We note that since leaves are not problematic, removing $R$ essentially boils down to removing the columns with index in $S$. Note that the remaining connected components are cycles of length 4. Cycles of length 4 can be always $3$-recolored.

Note that the above proof yields in fact an $O(\log n)$ rounds algorithm that outputs an $O(1)$ schedule. We can improve it into an $O(1)$-round algorithm, simply by pointing out that there is only a finite number of possible colorings for a column, and two adjacent columns cannot have the same coloring. This allows us to compute $S$ in constant time. 
\end{proof}

 \subparagraph{Non-existence.}

Let us now prove the negative result. Our high-level plan is as follows. Let $G$ be an $h \times w$ toroidal grid. We will look at all \emph{tiles} of size $2 \times 2$. If $G$ is properly colored with $k$ colors, so is each tile. The following two tiles are of special importance to us; we call these tiles of \emph{type A}:
\[
\begin{bmatrix}
\B & \C \\
\C & \A
\end{bmatrix},
\quad
\begin{bmatrix}
\A & \C \\
\C & \B
\end{bmatrix}.
\]
We are interested in the \emph{number} of type-A tiles. For example, consider the following colorings of the $3 \times 3$ toroidal grid:
\[
s = \begin{bmatrix}
\A & \B & \C \\
\B & \C & \A \\
\C & \A & \B
\end{bmatrix},
\quad
t = \begin{bmatrix}
\C & \A & \B \\
\B & \C & \A \\
\A & \B & \C
\end{bmatrix}.
\]
Here $s$ contains $3$ tiles of type A (recall that we wrap around at the boundaries), while $t$ does not have any tiles of type A. In particular, $s$ has an odd number of type-A tiles and $t$ has an even number of type-A tiles. In brief, we say that the \emph{A-parity} of $s$ is odd and the A-parity of $t$ is even. It turns out that this is sufficient to show that recoloring from $s$ to $t$ with one extra color is not possible (see the full version of the article for the proof of this lemma):

\begin{lem}\label{lem:gridsAparity}
    Let $G$ be a toroidal grid, and let $s$ and $t$ be two $3$-colorings. If $s$ and $t$ have different A-parities, then it is not possible to recolor $G$ from $s$ to $t$ with $1$ extra color.
\end{lem}

\begin{proof}
Let us now return to the last missing ingredient: the proof of Lemma~\ref{lem:gridsAparity}, which says that $3+1$ recoloring from $s$ to $t$ is not possible if we have different A-parities.

Intuitively, we would like to prove that recoloring operations preserve the A-parity. Unfortunately, this is not the case, as we can use the extra color. For example, if you take a proper $3$-coloring $s$ with an odd A-parity and replace all nodes of color $3$ with color $4$, you will have a proper $4$-coloring $t$ with an even A-parity.

It turns out that recoloring operations do preserve a certain kind of parity of $2 \times 2$ tiles, but it is much more involved than merely preserving type-A parity. We introduce a new set of $2 \times 2$ tiles that we call \emph{type-B} tiles; the parity of the number of type-B tiles is called {B-parity}:
\begin{align*} &
\begin{bmatrix}
\B & \A \\ \A & \D
\end{bmatrix},\quad
\begin{bmatrix}
\C & \A \\ \A & \D
\end{bmatrix},\quad
\begin{bmatrix}
\B & \A \\ \C & \D
\end{bmatrix},\quad
\begin{bmatrix}
\B & \C \\ \A & \D
\end{bmatrix},\quad
\begin{bmatrix}
\A & \C \\ \D & \B
\end{bmatrix},\quad
\begin{bmatrix}
\C & \B \\ \A & \D
\end{bmatrix},\quad
\begin{bmatrix}
\B & \C \\ \C & \D
\end{bmatrix},\\ &
\begin{bmatrix}
\D & \A \\ \A & \B
\end{bmatrix},\quad
\begin{bmatrix}
\D & \A \\ \A & \C
\end{bmatrix},\quad
\begin{bmatrix}
\D & \A \\ \C & \B
\end{bmatrix},\quad
\begin{bmatrix}
\D & \C \\ \A & \B
\end{bmatrix},\quad
\begin{bmatrix}
\B & \C \\ \D & \A
\end{bmatrix},\quad
\begin{bmatrix}
\D & \B \\ \A & \C
\end{bmatrix},\quad
\begin{bmatrix}
\D & \C \\ \C & \B
\end{bmatrix}.
\end{align*}

The collection of type-B tiles looks indeed somewhat arbitrary, but the following property is easy to verify: there is exactly one node of color $4$ in each type-A tile. Therefore if $x$ is a proper $3$-coloring, then B-parity of $x$ is even. In particular, in $3+1$ recoloring, the initial coloring $s$ and the target coloring $t$ both have even B-parities.

The magic behind the choice of type-B tiles is that it happens to satisfy the following property.

\begin{lem}\label{lem:ABparity}
If you change the color of one node in a properly $4$-colored grid so that the result is also a proper $4$-coloring, A-parity changes if and only if B-parity changes.
\end{lem}

\begin{proof}
Enumerate all $3\times 3$ neighborhoods and all possible ways to change the middle node, and check that the claim holds. We have made a computer program that verifies the claim and a human-readable list of all cases available online.\footnote{\url{https://github.com/suomela/recoloring}}
\end{proof}

Now any $3+1$ recoloring from $s$ to $t$ can be serialized so that we change the color of one node at a time. We start with an even B-parity (no type-B tiles in $s$), apply Lemma~\ref{lem:ABparity} at each step, and eventually we arrive at even B-parity (no type-B tiles in $t$). As B-parity did not change between $s$ and $t$, also A-parity cannot change between $s$ and $t$. This concludes the proof of Lemma~\ref{lem:gridsAparity}.
\end{proof}

Hence the A-parity of a coloring partitions the space of colorings in two components that are not connected by $3+1$ recoloring operations. To complete the proof of Theorem~\ref{thm:grids}, it now suffices to construct a pair of $3$-colorings with different A-parities for each relevant combination of $h$ and $w$.

\subparagraph{\boldmath Odd $h$, odd $w$.} First assume that both $h$ and $w$ are odd. The simplest case is $h = w$. In that case we can simply have $3$s on the anti-diagonal and color all remaining areas with colors $1$ and $2$; this gives a coloring $s$ with $h$ type-A tiles (odd type-A parity). If we put $3$s on the diagonal, we can construct a coloring $t$ with $0$ type-A tiles (even type-A parity). Here are examples for $h = w = 5$:
\[
s = \begin{bmatrix}
\A & \B & \A & \B & \C \\
\B & \A & \B & \C & \A \\
\A & \B & \C & \A & \B \\
\B & \C & \A & \B & \A \\
\C & \A & \B & \A & \B 
\end{bmatrix},
\quad
t = \begin{bmatrix}
\C & \A & \B & \A & \B \\
\B & \C & \A & \B & \A \\
\A & \B & \C & \A & \B \\
\B & \A & \B & \C & \A \\
\A & \B & \A & \B & \C 
\end{bmatrix}.
\]
If $h \ne w$, assume w.l.o.g.\ that $h < w$, and in particular, $w = h + 2\ell$ for some $\ell$. Then we can take the diagonal construction for $h \times h$ and add $\ell$ copies of the two rightmost columns. For example, for $h = 5$ and $w = 9$ (and hence $\ell = 2$) we get
\[
s = \begin{bmatrix}
\A & \B & \A & \B & \C & \B & \C & \B & \C \\
\B & \A & \B & \C & \A & \C & \A & \C & \A \\
\A & \B & \C & \A & \B & \A & \B & \A & \B \\
\B & \C & \A & \B & \A & \B & \A & \B & \A \\
\C & \A & \B & \A & \B & \A & \B & \A & \B 
\end{bmatrix},
\quad
t = \begin{bmatrix}
\C & \A & \B & \A & \B & \A & \B & \A & \B \\
\B & \C & \A & \B & \A & \B & \A & \B & \A \\
\A & \B & \C & \A & \B & \A & \B & \A & \B \\
\B & \A & \B & \C & \A & \C & \A & \C & \A \\
\A & \B & \A & \B & \C & \B & \C & \B & \C 
\end{bmatrix}.
\]
Note that each additional pair of columns results in one new tile of type A in both $s$ and $t$; hence overall we will have $h + \ell$ type-A tiles in $s$ and $\ell$ type-A tiles in $t$, and as $h$ was odd, the parity differs.

\subparagraph{\boldmath Odd $h$, even $w \ne 4$.} The case that remains to be considered is that exactly one of $h$ and $w$ is odd; w.l.o.g., assume that $h$ is odd and $w$ is even. Also recall that $w \ne 4$, and hence we can focus on the case $h \ge 3$ and $w \ge 6$.

For the base case of $h = 3$ and $w = 6$ we can use the following configuration; here in $s$ there is a sequence of $3$s that \emph{wraps around vertically twice}, while in $t$ there is a sequence of $3$s that does not wrap around vertically. Here $s$ has got $6$ type-A tiles (even), while $t$ has got $3$ type-A tiles (odd):
\[
s = \begin{bmatrix}
\A & \B & \C & \A & \B & \C \\
\B & \C & \A & \B & \C & \A \\
\C & \A & \B & \C & \A & \B 
\end{bmatrix},
\quad
t = \begin{bmatrix}
\A & \B & \A & \B & \A & \B \\
\C & \A & \C & \A & \C & \A \\
\B & \C & \B & \C & \B & \C 
\end{bmatrix}.
\]

If $h = 3$ and $w = 6 + 2\ell$, we can take the above construction and pad it by duplicating the two leftmost columns $\ell$ times. Each such duplication results in one new type-A tile in both configurations, maintaining the difference in parities. For example, for $h = 3$ and $w = 8$ we get
\[
s = \begin{bmatrix}
\A & \B & \A & \B & \C & \A & \B & \C \\
\B & \C & \B & \C & \A & \B & \C & \A \\
\C & \A & \C & \A & \B & \C & \A & \B 
\end{bmatrix},
\quad
t = \begin{bmatrix}
\A & \B & \A & \B & \A & \B & \A & \B \\
\C & \A & \C & \A & \C & \A & \C & \A \\
\B & \C & \B & \C & \B & \C & \B & \C 
\end{bmatrix}.
\]
Finally, if $h = 3 + 2\ell$, we can take the above construction for $h = 3$ and take $2\ell$ copies of the top row, shifting it back and forth to preserve the coloring. For example, for $h = 7$ and $w = 8$ we get
\[
s = \begin{bmatrix}
\A & \B & \A & \B & \C & \A & \B & \C \\
\B & \A & \B & \C & \A & \B & \C & \A \\
\A & \B & \A & \B & \C & \A & \B & \C \\
\B & \A & \B & \C & \A & \B & \C & \A \\
\A & \B & \A & \B & \C & \A & \B & \C \\
\B & \C & \B & \C & \A & \B & \C & \A \\
\C & \A & \C & \A & \B & \C & \A & \B 
\end{bmatrix},
\quad
t = \begin{bmatrix}
\A & \B & \A & \B & \A & \B & \A & \B \\
\B & \A & \B & \A & \B & \A & \B & \A \\
\A & \B & \A & \B & \A & \B & \A & \B \\
\B & \A & \B & \A & \B & \A & \B & \A \\
\A & \B & \A & \B & \A & \B & \A & \B \\
\C & \A & \C & \A & \C & \A & \C & \A \\
\B & \C & \B & \C & \B & \C & \B & \C 
\end{bmatrix}.
\]
This way we preserve the basic topological structure, with a sequence of $3$s wrapping around vertically twice in $s$ and zero times in $t$. Note that in $s$ we will get $2\ell$ new type-A tiles (as there are $2$ nodes of color $3$ in the top row) and in $t$ we will get $0$ new type-A tiles, again preserving the parity difference. This concludes the proof of Theorem~\ref{thm:grids}.

\section{Simple corollaries}\label{sec:simplecor}

\begin{lem}\label{lem:MISplusforest}
Assume that we are given a graph $G$ and input and target colorings with $k \ge 3$ colors. Assume that in $O(f(n))$ rounds we can find an independent set $I$ of $G$ such that $V \setminus I$ induces a forest of trees of depth at most $O(d(n))$. Then in $O(f(n)+d(n))$ rounds we can solve $k+1$ recoloring, with a schedule of length $O(d(n))$.
\end{lem}
\begin{proof}
Each node in $I$ switches to color $k+1$. We then use the algorithm described in the proof of Lemma \ref{lem:recoltree} to find a recoloring with schedule of length $O(d(n))$ for each connected component after the removal of $I$. After that, each node of $I$ can switch to its final color.
\end{proof}
\subparagraph{Acknowledgments.}

The authors would like to thank Nicolas Bousquet for helpful discussions regarding the proof of Lemma~\ref{lem:thereisthetaor2}.

\begin{table}[p]
\caption{Results: distributed recoloring in cycles (C) and paths (P).}
\label{tab:cycles}
\centering
\begin{tabular}{llllll}
\toprule
graph & input & extra & schedule & communication & reference \\
family & colors & colors & length & rounds & \\
\midrule
C/P & 2 & 0 & $\infty$ & & Lemma~\ref{lem:needsextra} \\
C/P & 2 & 1 & $\cO(1)$ & $0$ & Lemma~\ref{lem:minusone} \\
\midrule
C & 3 & 0 & $\infty$ & & Lemma~\ref{lem:needsextra} \\
P & 3 & 0 & $\Theta(n)$ & $\Theta(n)$ & Lemmas \ref{lem:3paths} and \ref{lem:3pathslb} \\
C/P & 3 & 1 & $\cO(1)$ & $\cO(1)$ & Lemma~\ref{lem:3plus1paths} \\
C/P & 3 & 2 & $\cO(1)$ & $0$ & Lemma~\ref{lem:minusone}\\
\midrule
C/P & 4 & 0 & $\cO(1)$ & $\cO(1)$ & Lemmas \ref{lem:3plus1paths} and \ref{lem:beyonddelta} \\
C/P & 4 & 3 & $\cO(1)$ & $0$ & Lemma~\ref{lem:minusone}\\
\bottomrule
\end{tabular}
\end{table}

\begin{table}[p]
\caption{Results: distributed recoloring in $3$-regular trees.}
\label{tab:3regtree}
\centering
\begin{tabular}{lllll}
\toprule
input & extra & schedule & communication & reference \\
colors & colors & length & rounds & \\
\midrule
2 & 0 & $\infty$ & & Lemma~\ref{lem:needsextra} \\
2 & 1 & $\cO(1)$ & $0$ & Lemma~\ref{lem:minusone}\\
\midrule
3 & 0 & $\Theta(n)$ & $\Theta(n)$ & Lemmas \ref{lem:3paths} and \ref{lem:3pathslb} \\
3 & 1 & $\cO(1)$ & $\cO(\log n)$ & Theorem~\ref{thm:tree3-1}\\
3 & 2 & $\cO(1)$ & $0$ & Lemma~\ref{lem:minusone}\\
\midrule
4 & 0 & $\Theta(\log n)$ & $\Theta(\log n)$ & Theorem~\ref{thm:tree4} and Lemma~\ref{lem:4treelb} \\
4 & 1 & $\cO(1)$ & $\cO(1)$ & Lemma~\ref{lem:4plus1subcubic} \\
4 & 3 & $\cO(1)$ & $0$ & Lemma~\ref{lem:minusone}\\
\midrule
5 & 0 & $\cO(1)$ & $\cO(1)$ & Lemmas \ref{lem:4plus1subcubic} and \ref{lem:beyonddelta}\\
\bottomrule
\end{tabular}
\end{table}

\begin{table}[p]
\caption{Results: distributed recoloring in trees.}
\label{tab:trees}
\centering
\begin{tabular}{lllll}
\toprule
input & extra & schedule & communication & reference \\
colors & colors & length & rounds & \\
\midrule
2 & 0 & $\infty$ & & Lemma~\ref{lem:needsextra} \\
2 & 1 & $\cO(1)$ & $0$ & Lemma~\ref{lem:minusone}\\
\midrule
3 & 0 & $\Theta(n)$ & $\Theta(n)$ & Lemmas \ref{lem:3paths} and \ref{lem:3pathslb} \\
3 & 1 & $\cO(1)$ & $\cO(\log n)$ & Theorem~\ref{thm:tree3-1}\\
3 & 2 & $\cO(1)$ & $0$ & Lemma~\ref{lem:minusone}\\
\midrule
4 & 0 & $\Theta(\log n)$ & $\Theta(\log n)$ & Theorem~\ref{thm:tree4} and Lemma~\ref{lem:4treelb}\\
4 & 1 & $\cO(1)$ & $\cO(\log n)$ & Theorem~\ref{thm:tree3-1}\\
4 & 3 & $\cO(1)$ & $0$ & Lemma~\ref{lem:minusone}\\
\bottomrule
\end{tabular}
\end{table}

\begin{table}[p]
\caption{Results: distributed recoloring in toroidal grids. The distributed complexity of $4+1$ recoloring is left as an open question. However, by prior work it is known that $4+1$ recoloring is always possible: grids are $4$-regular graphs, therefore they are $4$-recolorable with Kempe operations, and hence also with $1$ extra color.}
\label{tab:grids}
\centering
\begin{tabular}{lllll}
\toprule
input & extra & schedule & communication & reference \\
colors & colors & length & rounds & \\
\midrule
2 & 0 & $\infty$ & & Lemma~\ref{lem:needsextra}\\
2 & 1 & $\cO(1)$ & $0$ & Lemma~\ref{lem:minusone}\\
\midrule
3 & 0 & $\infty$ & & Lemma~\ref{lem:needsextra} \\
3 & 1 & $\infty$ & & Theorem~\ref{thm:grids} \\
3 & 2 & $\cO(1)$ & $0$ & Lemma~\ref{lem:minusone}\\
\midrule
4 & 0 & $\infty$ & & Lemma~\ref{lem:needsextra} \\
4 & 1 & ? & ? \\
4 & 2 & $\cO(1)$ & $\cO(1)$ & Lemma~\ref{lem:4plus2grids} \\
4 & 3 & $\cO(1)$ & $0$ & Lemma~\ref{lem:minusone}\\
\midrule
5 & 0 & $\infty$ & & Lemma~\ref{lem:needsextra} \\
5 & 1 & $\cO(1)$ & $\cO(1)$ & Lemma~\ref{lem:5plus1grids} \\
5 & 4 & $\cO(1)$ & $0$ & Lemma~\ref{lem:minusone}\\
\midrule
6 & 0 & $\cO(1)$ & $\cO(1)$ & Lemmas \ref{lem:5plus1grids} and \ref{lem:beyonddelta} \\
\bottomrule
\end{tabular}
\end{table}

\begin{table}
\caption{Results: distributed recoloring in subcubic graphs.}
\label{tab:3reg}
\centering
\begin{tabular}{lllll}
\toprule
input & extra & schedule & communication & reference \\
colors & colors & length & rounds & \\
\midrule
2 & 0 & $\infty$ & & Lemma~\ref{lem:needsextra} \\
2 & 1 & $\cO(1)$ & $0$ & Lemma~\ref{lem:minusone}\\
\midrule
3 & 0 & $\infty$ & & Lemma~\ref{lem:needsextra} \\
3 & 1 & $\cO(\log n)$ & $\cO(\log^2 n)$ & Theorem~\ref{thm:cubic3-1}  \\
3 & 2 & $\cO(1)$ & $0$ & Lemma~\ref{lem:minusone}\\
\midrule
4 & 0 & $\infty$ & & Lemma~\ref{lem:needsextra} \\
4 & 1 & $\cO(1)$ & $\cO(1)$ & Lemma~\ref{lem:4plus1subcubic} \\
4 & 3 & $\cO(1)$ & $0$ & Lemma~\ref{lem:minusone}\\
\midrule
5 & 0 & $\cO(1)$ & $\cO(1)$ & Lemma~\ref{lem:beyonddelta}  \\
\bottomrule
\end{tabular}
\end{table}

\clearpage

\bibliographystyle{plainurl}
\bibliography{Bibliography}

\end{document}